\documentclass{amsart} 
\usepackage[margin=1.2in]{geometry}

\usepackage{float}

\usepackage{versions}  

\usepackage[svgnames]{xcolor}

\usepackage{natbib, amsfonts, amsthm}
\usepackage{adjustbox}

\usepackage{thmtools}
\usepackage{thm-restate}

\usepackage[ruled,vlined,linesnumbered,resetcount]{algorithm2e}
\usepackage{multirow}
\newcommand{\vertiii}[1]{{\left\vert\kern-0.25ex\left\vert\kern-0.25ex\left\vert #1
		\right\vert\kern-0.25ex\right\vert\kern-0.25ex\right\vert}}

\newtheorem{lemma}{Lemma}
\newtheorem{theorem}{Theorem}

\newcommand{\VV}{\mathrm{var}}
\newcommand{\EE}{\mathbb{E}}
\newcommand{\PP}{\mathbb{P}}

\newcommand{\PA}{\mathrm{pa}}
\newcommand{\CH}{\mathrm{ch}}
\newcommand{\AN}{\mathrm{an}}
\newcommand{\DE}{\mathrm{de}}
\newcommand{\G}{\mathcal{G}}
\newcommand{\Sigmah}{\widehat{\Sigma}}

\title{On Causal Discovery with Equal Variance Assumption}

\author{Wenyu Chen}
\address{Department of Statistics, University
	of Washington}
\email{wenyuc@uw.edu}
\author{Mathias Drton}
\address{Department of Mathematical Sciences, University
	of Copenhagen}
\email{md5@uw.edu}
\author{Y.~Samuel Wang}
\address{Booth School of Business, The University of Chicago}
\email{swang24@uchicago.edu}

\begin{document}

	\begin{abstract}
		Prior work has shown that causal structure can be uniquely
		identified from observational data when these follow a structural
		equation model whose error terms have equal variances.  We show that
		this fact is implied by an ordering among (conditional) variances.
		We demonstrate that ordering estimates of these variances yields a
		simple yet state-of-the-art method for causal structure learning
		that is readily extendable to high-dimensional problems.
	\end{abstract}

\maketitle	
	
	\section{Introduction}
	A structural equation model for a random vector
	$X = (X_1, \ldots, X_p)$ postulates causal relations in which each
	variable $X_j$ is a function of a subset of the other
	variables and a stochastic error $\varepsilon_j$.  Causal
	discovery/structure learning is the problem of inferring which of
	other variables each variable $X_j$ depends on.  We consider this
	problem 
        where only observational
	data, that is, a sample from the joint distribution of $X$, is
	available.  While in general only an equivalence class of structures can
	then be inferred \citep{pearl:2009,spirtes2000causation}, recent work
	stresses that unique identification is possible under assumptions such
	as non-linearity with additive errors, linearity with non-Gaussian
	errors, and linearity with errors of equal variance; see the reviews
	of \citet{drton:2017} and \citet{heinze:2018} or the book of
	\citet{jonas:book}.
	
	This note is concerned with the equal variance case
	treated by \citet{Peters2014} and \citet{Loh2014} who prove
	identifiability of the causal structure and propose greedy search
	methods for its estimation.  Our key observation is that the
	identifiability is implied by an ordering among certain conditional
	variances.  Ordering estimates of these variances yields a fast
	method for estimation of the causal ordering of the variables.  The
	precise causal structure can then be inferred using variable selection
	techniques for regression \citep{shojaie:2010}.  Specifically,
	we develop a top-down approach that infers the
	ordering by successively identifying sources. The method is
	developed for low- as well as high-dimensional problems.
	Simulations	show significant gains in computational efficiency when compared with
	greedy search and increased accuracy when the number of variables
	$p$ is large.
	
	An earlier version of this note also included a bottom-up
        method which identified the causal ordering by successively
        finding sinks via minimal precisions.  However, after the
        note was finished, we became aware of~\citet{Ghoshal2018} who
        proposed a similar bottom-up approach.  We emphasize that our
        top-down approach only requires control of the maximum
        in-degree as opposed to the bottom-up approach which requires
        control of the maximum Markov blanket. This is discussed
        further in Section~\ref{sec:high-dimens-probl} and a direct
        numerical comparison is given in
        Section~\ref{sec:highDSimulations}.
	
	\section{Structural Equation Models and Directed Acyclic Graphs}
	
	Suppose, without loss of generality, that the observed random vector
	$X=(X_1,\dots,X_p)$ is centered.  In a linear structural equation
	model, $X$ then solves an equation system
	\begin{equation}\label{eq:SEM}
	X_j = \sum_{k\not= j}\beta_{jk}X_k + \varepsilon_j,\qquad j=1,\ldots, p,
	\end{equation}
	where the $\varepsilon_j$ are independent random variables with mean zero, and
	the coefficients $\beta_{jk}$ are unknown parameters.  Following
	\cite{Peters2014}, we assume that all $\varepsilon_j$ have a common
	unknown variance $\sigma^2>0$.  We will write $X\sim (B,\sigma^2)$ to
	express the assumption that there indeed exist independent errors
	$\varepsilon_1,\dots,\varepsilon_p$ of equal variance $\sigma^2$ such that $X$
	solves~(\ref{eq:SEM}) for coefficients given by a real $p\times p$
	matrix $B=(\beta_{jk})$ with zeros along the diagonal.
	
	The causal structure inherent to the equations in~(\ref{eq:SEM}) is
	encoded in a directed graph $\G(B)$ with vertex set $V=\{1,\dots,p\}$
	and edge set $E(B)$ equal to the support of $B$.  So,
	$E(B)=\{(k,j): \beta_{jk}\not=0\}$.  Inference of $\G(B)$ is the goal
	of causal discovery as considered in this paper.  As in related work,
	we assume $\G(B)$ to be a directed acyclic graph (DAG) so that $B$ is permutation similar to a triangular matrix.
	Then~(\ref{eq:SEM}) admits the unique solution $X=(I-B)^{-1}\varepsilon$ where
	$\varepsilon=(\varepsilon_1,\dots,\varepsilon_p)$.  Hence, the covariance matrix of
	$X\sim (B,\sigma^2)$ is
	\begin{equation}
	\label{eq:Sigma}
	\Sigma := \EE(XX^T) = \sigma^2(I- B)^{-1}(I-B)^{-T}.
	\end{equation}
	
	We will invoke the following graphical concepts.
	If the considered graph $\G$ contains
	the edge $k\to j$, then $k$ is a parent of its child
	$j$.  We write $\PA(j)$ for the set of all parents of a node $j$.
	Similarly, $\CH(j)$ is the set of children of $j$.  If there exists a
	directed path $k \rightarrow \ldots \rightarrow j$, then $k$ is an
	ancestor of its descendant $j$.  The sets of ancestors
	and descendants of $j$ are $\AN(j)$ and $\DE(j)$, respectively.  Here,
	$j\in\AN(j)$ and $j\in\DE(j)$.  A set of nodes $C$ is ancestral
	if $\AN(j)\subseteq C$ for all $j\in C$.
	If $\G$ is a DAG, then it admits a topological ordering of its
	vertices.  In other words, there exists a numbering $\sigma$ such that
	$\sigma(j) < \sigma(k)$ only if $k \notin \AN(j)$.  Finally, every DAG
	contains at least one source, that is, a node $j$ with
	$\PA(j)=\emptyset$.  Similarly, every DAG contains at least one sink,
	which is a node $j$ with $\CH(j)=\emptyset$.

	
	\section{Identifiability by Ordering Variances}
	\label{sec:ident-from-vari}

	The main result of \citet{Peters2014} shows that the graph
	$\G(B)$ and the parameters $B$ and $\sigma^2$ are identifiable
	from the covariance in~(\ref{eq:Sigma}). No faithfulness
	assumptions are needed.

	\begin{theorem}
		\label{thm:peters}
		Let $X\sim(B_X,\sigma_X^2)$ and $Y\sim(B_Y,\sigma_Y^2)$ with both
		$\G(B_X)$ and $\G(B_Y)$ directed and acyclic.  If
		$\VV(X)=\VV(Y)$, then $\G(B_X)=\G(B_Y)$, $B_X=B_Y$, and
		$\sigma_X^2=\sigma_Y^2$.
	\end{theorem}
	
	In this section we first give an inductive proof of
	Theorem~\ref{thm:peters} that proceeds by recursively identifying
	source nodes for $\G(B)$ and subgraphs.  We then clarify that
	alternatively one could identify sink nodes.
	Our first lemma clarifies that the sources in $\G(B)$ are
	characterized by minimal variances.
	We define
	\begin{equation}\label{eq:betaMin}
	\zeta\equiv \zeta(B) \;=\; \min_{(k,j)\in E(B)} \beta_{jk}^2.
	\end{equation}

	\begin{lemma}\label{lem:cov_known}
		Let $X \sim (B, \sigma^2)$ with $\G(B)$ directed and acyclic.
		If $\PA(j)=\emptyset$, then $\VV(X_j)=\sigma^2$.  If
		$\PA(j) \neq \emptyset$, then
		$\VV(X_j)\ge \sigma^2(1+\zeta)>\sigma^2$.
	\end{lemma}
	\begin{proof}
		Let $\Pi=(\pi_{jk})=(I-B)^{-1}$.  Each total
			effect $\pi_{jk}$ is the sum over directed paths from $k$
		to $j$ of products of coefficients $\beta_{ab}$ along each path.  In
		particular, $\pi_{jj}=1$.
		From~(\ref{eq:Sigma}),
		$\VV(X_j)=\sigma^2 \sum_{k=1}^p\pi_{jk}^2$.  Hence, if
		$\PA(j) = \emptyset$, then $\VV(X_j) = \sigma^2$ because
		$\pi^2_{jk} = 0$ for all $k \neq j$.  If $\PA(j)\not=\emptyset$ then
		by acyclicity of $\G(B)$
		there exists a node $\ell \in \PA(j)$ such that
		$\DE(\ell) \cap \PA(j) = \{\ell\}$.
		Then $\pi_{j\ell}^2 = \beta_{j\ell}^2 \ge \zeta$ and
		\begin{equation*}\label{eq:cov_diag2}
		\VV(X_j) =
		\sigma^2\bigg( 1+ \sum_{k\not=j}\pi_{jk}^2 \bigg)
		\ge
		\sigma^2\left( 1+ \pi_{j\ell}^2 \right) \ge
		\sigma^2\left( 1+ \zeta \right).
		\end{equation*}
	\end{proof}

	The next lemma shows that by conditioning on a source, or more
	generally an ancestral set, one recovers a structural equation model with equal error
	variance whose graph has the source node or the entire ancestral set
	removed.  For a variable $X_j$ and a vector $X_C=(X_k:k\in C)$, we
	define $X_{j.C}=X_j-\EE(X_j|X_C)$.

	\begin{lemma}\label{lem:adjust}
		Let $X \sim(B, \sigma^2)$ with $\G(B)$ directed and acyclic. Let $C$ be
		an ancestral set in $\G(B)$.  Then
		$(X_{j. C}:j\notin C) \sim (B[-C], \sigma^2)$ for submatrix
		$B[-C]=(\beta_{jk})_{j,k\notin C}$.
	\end{lemma}
	\begin{proof}
		Let $j\notin C$.  Since $C$ is ancestral, $X_C$ is a function of
		$\varepsilon_C$ only and thus independent of $\varepsilon_j$.  Hence,
		$\EE(\varepsilon_j|X_C)=\EE(\varepsilon_j)=0$.  Because it also holds that $X_{k.C}=0$
		for $k\in C$, we have from~(\ref{eq:SEM}) that
		\[
		X_{j.C} = \sum_{k\in\PA(j) \setminus C}\beta_{jk}X_{k.C} + \varepsilon_j.\]
	\end{proof}
	
	The lemmas can be combined to identify a topological
	ordering of $\G(B)$ and prove Theorem~\ref{thm:peters}.
	
	\begin{proof}[of Theorem~\ref{thm:peters}]
		The claim is trivial for $p=1$ variables, which gives the base for
		an induction on $p$.  If $p>1$, then Lemma~\ref{lem:cov_known}
		identifies a source $c$ by variance minimization.  Conditioning
		on $c$ as in Lemma~\ref{lem:adjust} reduces the problem to
		size $p-1$.  By the induction assumption, $\sigma^2$ and
		$B[-\{c\}]$ can be identified.  The regression coefficients in
		the conditional expectations $\EE\left(X_j|X_c\right)$ for $j\not=c$
		identify the missing first row and column of $B$; see
		e.g.~\citet[\S7]{D}.
	\end{proof}
	
	Next, we show that alternatively one may minimize precisions to identify a
	sink node. We state analogues of Lemma~\ref{lem:cov_known} and \ref{lem:adjust} which can also be used to prove Theorem~\ref{thm:peters}. 

	\begin{lemma}\label{lem:prec_diag}
		Let $X \sim (B, \sigma^2)$ with $\G(B)$ directed and acyclic.  Let
		$\Sigma$ be the covariance matrix of $X$, and $\Phi=\Sigma^{-1}$ the
		precision matrix.   If $\CH(j)=\emptyset$, then
		$\Phi_{jj}=1/\sigma^2$.  If $\CH(j)\not=\emptyset$, then
		$\Phi_{jj}\ge \{1 + \zeta|\CH(j)|\}/\sigma^2 >1/\sigma^2$.
	\end{lemma}
	\begin{proof}
		The diagonal entries of $\Phi =\frac{1}{\sigma^2} (I-B) (I-B)^T$
		are
		$
		\Phi_{jj}=
		\frac{1}{\sigma^2} (1+\sum_{k\in \CH(j)}\beta_{kj}^2).
		$
		So $\Phi_{jj} = 1/\sigma^2$ if $\CH(j) = \emptyset$, and
		$\Phi_{jj} \ge \{1+|\CH(j)|\zeta\}/\sigma^2$ if
		$\CH(j) \neq \emptyset$.
	\end{proof}
	
	Marginalization of a sink is justified by the following well-known
	fact \cite[e.g.][\S5]{D}.
	
	\begin{lemma}\label{lem:prec_known}
		Let $X \sim(B, \sigma^2)$ with $\G(B)$ directed and acyclic.  Let $C$ be
		an ancestral set in $\G(B)$.  Then
		$X_C \sim (B[C], \sigma^2)$ for submatrix $B[C]=(\beta_{jk})_{j,k\in C}$.
	\end{lemma}

	\section{Estimation Algorithms}
	\label{sec:low-dimens-algor}
	
	\subsection{Low-dimensional Problems}
	\label{sec:low-dimens-probl}
	The results from Section~\ref{sec:ident-from-vari} naturally
        yield an iterative top-down algorithm for estimation of a topological
        ordering for $\G(B)$. In each
        step of the procedure we select a source node by
        comparing variances conditional on the previously selected
        variables, so the criterion in the minimization
        in 
        Algorithm~\ref{alg:source} is the variance
	\begin{equation}\label{eq:sourceCriterion}
	f_1(\hat\Sigma, \Theta, j)
	\;=\;\hat\Sigma_{j,j}-\hat\Sigma_{j,\Theta}\hat\Sigma_{\Theta,\Theta}^{-1}\hat\Sigma_{\Theta,j}\;=\;
	\frac{1}{\{(\hat\Sigma_{\Theta\cup \{j\},\Theta\cup\{j\}})^{-1}\}_{j,j}},
	\end{equation}
	where $\hat\Sigma$ is the sample covariance matrix.
        Alternatively, and as also observed by \citet{Ghoshal2018}, a
        bottom-up procedure could construct the reverse causal
        ordering by successively minimizing precisions (or in other
        words, full conditional variances).

	\begin{algorithm}[t]
		\label{alg:source}
		\caption{Topological Ordering: General procedure with criterion $f$}
		\SetKwInOut{Input}{Input}
		\SetKwInOut{Output}{Output}
		\Input{$\hat \Sigma \in\mathbb{R}^{p\times p}$ (estimated) covariance of $X$\\}
		\Output{$\Theta$ }
		$\Theta^{(0)}\gets \emptyset$\;
		\For{$z=1,\ldots,p$}{
			$\theta \gets \arg\min_{j \in V \setminus \Theta^{(z-1)}} f(\hat \Sigma, \Theta^{(z-1)}, j)$\; \label{algLine:criteria}
			Append $\theta$ to $\Theta^{(z-1)}$ to form $\Theta^{(z)}$
		}
		\Return the ordered set $\Theta^{(p)}$.
	\end{algorithm}

	To facilitate theoretical statements about our top-down procedure, we
	 assume that the errors $\varepsilon_j$ in~(\ref{eq:SEM}) are all
	sub-Gaussian with maximal sub-Gaussian parameter $\gamma>0$.
	We indicate this by writing $X\sim (B,\sigma^2,\gamma)$.
	Our analysis is restricted to inference of a topological ordering.
	\cite{shojaie:2010} give results on lasso-based inference of the graph given an ordering.
	
	\begin{theorem}\label{thm:topDownProof}
		Let $X \sim (B, \sigma^2,\gamma)$ with $\G(B)$ directed and acyclic.
		Suppose the covariance matrix $\Sigma=\EE(XX^T)$ has minimum eigenvalue $\lambda_{\min}>0$.
		If
		\[
		n\;> \;
		p^2\left\{\log(p^2+p)-\log\left(\epsilon/2\right)\right\}128\left(1+4\frac{\gamma^2}{\sigma^2}\right)^2
		\left(\max_{j\in V}
		\Sigma_{j,j}\right)^2\left(\frac{ \zeta \lambda_{\min} +
			2\sigma^2}{\zeta\lambda_{\min}^2  }  \right)^2,
		\]
		then Algorithm~\ref{alg:source} using criterion criterion~\eqref{eq:sourceCriterion} recovers a topological ordering of $\G(B)$ with probability at least
$1-\epsilon$.
	\end{theorem}
	
	The result follows using concentration for sample covariances
	\citep[Lemma 1]{ravikumar2011} and error propagation analysis as in
	\citet[Lemma 5]{Harris2013}.  We give details in
	Appendix~\ref{sec:error-prop-proofs}, which is found in the
        supplementary materials.

	\subsection{High-dimensional Problems}
	\label{sec:high-dimens-probl}
	
	The consistency result in Theorem~\ref{thm:topDownProof} requires the
	sample size $n$ to exceed a multiple of $p^2\log(p)$ and only
	applies to low-dimensional problems.   If $p>n$, method will stop at the $n$th step when the conditional variance in~(\ref{eq:sourceCriterion}) becomes zero for
	all $j\notin \Theta$.

	However, in the high-dimensional setting if $\G(B)$ has maximum in-degree bounded
	by a small integer $q$, we may modify the criterion
	from~(\ref{eq:sourceCriterion}) to
	\begin{equation}\label{eq:sourceHD}
	f_2(\hat \Sigma, \Theta, j) \;=\; \min_{C \subseteq \Theta, |C| =
		q}f_1(\hat \Sigma,C,j) \;=\; \min_{C \subseteq \Theta, |C| = q} \hat \Sigma_{j,j} - \hat \Sigma_{j, C}(\hat \Sigma_{C, C})^{-1} \hat \Sigma_{C, j}.
	\end{equation}
	The intuition is that in the population case, adjusting by a smaller
	set $C \subseteq \Theta^{(z)}$ with $\PA(j)\subseteq C$ yields the same
	results as adjusting by all of $\Theta^{(z)}$.  The next lemma makes
	the idea rigorous.
	
	\begin{lemma}\label{lem:topDownHD}
		Let $X \sim (B, \sigma^2)$ with $\G(B)$ directed and acyclic with
		maximum in-degree at most $q$.  Let $\Sigma=\EE(XX^T)$, and suppose
		$S\subseteq V \setminus\{j\}$ is an ancestral set.  If
		$\PA(j)\subseteq S$, then $f_2(\Sigma, S, j)=\sigma^2$.  If
		$\PA(j)\not\subseteq S$, then
		$f_2(\Sigma, S, j)\ge \sigma^2(1 + \zeta)$.
	\end{lemma}
	\begin{proof}
		The conditional variance of $X_j$ given $X_S$ is the variance of the
		residual $X_{j.S}$.  By Lemma~\ref{lem:adjust}, $X_{j.S}$ has the
		same distribution as $X_j'$ when $X'\sim(B[-S],\sigma^2)$.  Now, $j$
		is a source of $\G(B[-S])$ if and only if
		$\PA(j)\subseteq S$.  Lemma~\ref{lem:cov_known} implies that
		$\VV(X_j|X_C)=\sigma^2$ if $\PA(j)\subseteq S$ and
		$\VV(X_j|X_C)\ge \sigma^2(1 + \zeta)$ otherwise.  The claim about
		$f_2(\Sigma,S,j)$ now follows.
	\end{proof}
	
	Based on Lemma~\ref{lem:topDownHD}, we have the following result whose
	proof is analogous to that of Theorem~\ref{thm:topDownProof}.  The key
	feature of the result is a drop from $p^2$ to $(q+1)^2$ in the
        sample size requirement.
	
	\begin{theorem}
		Let $X \sim ( B, \sigma^2,\gamma)$ with $\G(B)$ directed and acyclic with
		of maximum in-degree at most $q$.  Suppose all $(q + 1)\times (q + 1)$ principal
		submatrices of $\Sigma = \EE\left(XX^T\right)$ have minimum eigenvalue at least
		$\lambda_{\min}>0$.  
		If
		\[
		n\;>\; (q+1)^2\left\{\log(p^2+p)-\log\left(\epsilon/2\right)\right\}128\left(1+4\frac{\gamma^2}{\sigma^2}\right)^2
		\left(\max_{j\in V}
		\Sigma_{j,j}\right)^2\left(\frac{ \zeta \lambda_{\min} +
			2\sigma^2}{\zeta\lambda_{\min}^2  }  \right)^2,
		\]
		then Algorithm~\ref{alg:source} using criterion~\eqref{eq:sourceHD}
		recovers a topological ordering of $\G(B)$ with probability at least
		$1-\epsilon$.
	\end{theorem}

	We contrast our guarantees with those for the bottom-up method
        of \citet{Ghoshal2018} which selects sinks by minimizing
        conditional precisions that are estimated using the CLIME
        estimator \citep{Cai2011}.  Because CLIME requires small
        Markov blankets, the bottom-up procedure has sample complexity
        $\mathcal{O}\left(d^8 \log(p)\right)$ where $d$ is the maximum
        total degree. This implies that the procedure cannot
        consistently discover graphs with hubs, i.e., nodes with very
        large out-degree, in the high dimensional setting.  This said,
        the computational complexity of the bottom-up procedure is
        polynomial in $d$, while our top-down procedure is exponential
        in the maximum in-degree.  In practice, we use a
        branch-and-bound procedure \citep{lumley2017leaps} to
        efficiently select the set which minimizes the conditional
        variance; see Section~\ref{sec:highDSimulations}.

	\section{Numerical Results}
	\label{sec:simulations}

	\subsection{Low-dimensional Setting}\label{sec:low-dim-sim}
	
	We first assess performance in the low-dimensional setting. Random
	DAGs with $p$ nodes and a unique topological ordering are generated
	by: (1) always including edge $v \rightarrow v+1$ for $v < p$, and (2)
	including edge $v \rightarrow u$ with probability $p_c$ for all
	$v < u - 1$. We consider a sparse setting with $p_c=3/(2p-2)$ and a
	dense setting with $p_c=0.3$.  All linear coefficients
	are drawn uniformly from $\pm[.3, 1]$.  The error terms are standard
	normal.   
	Performance is measured using Kendall's $\tau$ 
	between rankings of variables according to 
	the true and estimated topological orderings.  
	Although the true graph admits a unique ordering by construction, the graph
	estimated by the greedy search may not admit a unique ordering. 
	Nevertheless, the ranking of variables according to the estimated 
	graph is unique if we allow ties, and Kendall's $\tau$ remains a good measure
	for all the methods. 
	We also compute the percentage of
	true edges discovered (Recall), the percentage of estimated edges that
	are flipped in the true graph (Flipped), and the proportion of
	estimated edges which are either flipped or not present in the true
	graph (false discovery rate; FDR). Tables~\ref{tab:chain_dense}
	and~\ref{tab:chain_sparse} show averages over 500 random
	realizations for our top-down procedure (TD), the
	bottom-up procedure (BU) of \citet{Ghoshal2018}, and greedy
        DAG search (GDS). For the bottom up
        procedure in the low-dimensional setting, we may in fact
        simply invert the sample covariance to estimate precisions. 
	For GDS, we allow for 5 random restarts using the
	same procedure as \citet{Peters2014}.
	
	\begin{table}[t]
	\centering
	\caption{Low-dimensional dense settings}
	\label{tab:chain_dense}
	\begin{tabular}{ll|c|c|c|c|c|c|c|c|c|c|c|c|}
		\cline{3-14}
		&      & \multicolumn{3}{c|}{Kendall's $\tau$} & \multicolumn{3}{c|}{Recall \%} & \multicolumn{3}{c|}{Flipped \% } & \multicolumn{3}{c|}{FDR \%} \\ \hline
		\multicolumn{1}{|l|}{$p$}                   & $n$    & TD& BU     & GDS        & TD& BU       & GDS       & TD& BU       & GDS        & TD& BU     & GDS   \\ \hline \hline
		\multicolumn{1}{|l|}{\multirow{3}{*}{5}}  & 100 & 0.85 & 0.82 & 0.88 & 91 & 89 & 91 & 7 & 8 & 6 & 17 & 18 & 9 \\  \cline{2-14}
		\multicolumn{1}{|l|}{}                    & 500  & 0.98 & 0.97 & 0.98 & 99 & 98 & 99 & 1 & 1 & 1 & 4 & 4 & 2 \\ \cline{2-14}
		\multicolumn{1}{|l|}{}                    & 1000 & 0.99 & 0.98 & 0.99 & 99 & 99 & 99 & 1 & 1 & 1 & 3 & 3 & 1 \\     \hline \hline
		\multicolumn{1}{|l|}{\multirow{3}{*}{20}} & 100  & 0.92 & 0.85 & 0.61 & 85 & 83 & 62 & 3 & 5 & 13 & 32 & 35 & 43 \\ \cline{2-14}
		\multicolumn{1}{|l|}{}                    & 500  & 0.99 & 0.97 & 0.75 & 99 & 98 & 81 & 1 & 1 & 11 & 28 & 29 & 35 \\ \cline{2-14}
		\multicolumn{1}{|l|}{}                    & 1000 & 1.00 & 0.99 & 0.82 & 100 & 100 & 88 & 0 & 0 & 8 & 26 & 26 & 28 \\ \hline \hline
		\multicolumn{1}{|l|}{\multirow{3}{*}{40}} & 100  & 0.96 & 0.91 & 0.53 & 71 & 69 & 44 & 2 & 3 & 11 & 41 & 43 & 58 \\  \cline{2-14}
		\multicolumn{1}{|l|}{}                    & 500  & 0.99 & 0.98 & 0.59 & 96 & 96 & 63 & 0 & 1 & 14 & 41 & 42 & 57 \\  \cline{2-14}
		\multicolumn{1}{|l|}{}                    & 1000  & 1.00 & 0.99 & 0.64 & 97 & 97 & 71 & 0 & 0 & 14 & 40 & 41 & 57 \\  \hline
	\end{tabular}
\end{table}

	\begin{table}[t]
	\centering
	\caption{Low-dimensional sparse settings}
	\label{tab:chain_sparse}
	\begin{tabular}{ll|c|c|c|c|c|c|c|c|c|c|c|c|}
		\cline{3-14}
		&      & \multicolumn{3}{c|}{Kendall's $\tau$} & \multicolumn{3}{c|}{Recall \%} & \multicolumn{3}{c|}{Flipped \% } & \multicolumn{3}{c|}{FDR \%} \\ \hline
		\multicolumn{1}{|l|}{$p$}                   & $n$    & TD& BU     & GDS        & TD& BU       & GDS       & TD& BU       & GDS        & TD& BU     & GDS   \\ \hline \hline
		\multicolumn{1}{|l|}{\multirow{3}{*}{5}}  & 100 & 0.87 & 0.84 & 0.88 & 91 & 89 & 90 & 6 & 7 & 6 & 16 & 17 & 9 \\  \cline{2-14}
		\multicolumn{1}{|l|}{}                    & 500  & 0.98 & 0.96 & 0.98 & 98 & 98 & 99 & 1 & 2 & 1 & 5 & 5 & 2 \\ \cline{2-14}
		\multicolumn{1}{|l|}{}                    & 1000 & 0.99 & 0.98 & 0.99 & 99 & 99 & 99 & 1 & 1 & 1 & 3 & 4 & 1 \\     \hline \hline
		\multicolumn{1}{|l|}{\multirow{3}{*}{20}} & 100  & 0.77 & 0.59 & 0.60 & 85 & 79 & 77 & 9 & 13 & 15 & 35 & 40 & 39 \\ \cline{2-14}
		\multicolumn{1}{|l|}{}                    & 500  & 0.96 & 0.88 & 0.77 & 98 & 96 & 89 & 2 & 4 & 10 & 19 & 22 & 26 \\ \cline{2-14}
		\multicolumn{1}{|l|}{}                    & 1000 & 0.99 & 0.94 & 0.81 & 100 & 98 & 90 & 0 & 2 & 9 & 14 & 16 & 23 \\ \hline \hline
		\multicolumn{1}{|l|}{\multirow{3}{*}{40}} & 100  & 0.72 & 0.44 & 0.47 & 81 & 72 & 72 & 10 & 16 & 20 & 38 & 46 & 54 \\  \cline{2-14}
		\multicolumn{1}{|l|}{}                    & 500  & 0.96 & 0.80 & 0.58 & 98 & 94 & 81 & 2 & 5 & 18 & 24 & 31 & 47 \\  \cline{2-14}
		\multicolumn{1}{|l|}{}                    & 1000  & 0.99 & 0.91 & 0.61 & 99 & 98 & 82 & 1 & 2 & 17 & 17 & 22 & 48 \\  \hline
	\end{tabular}
\end{table}

	In both dense and sparse settings, when $p = 5$,
 	greedy search performs best in all metrics.  However, for $p = 20$ and
	$40$, the top-down approach does best, followed by
	bottom-up, and finally greedy search. The top-down and bottom-up method both have a
	substantially higher average Kendall's $\tau$ than greedy search.  
	
	In our experiments, the proposed methods are roughly 50 to 500 times faster
	than greedy search as graph size and density increases.
	On our personal computer, the average run time in the dense setting
	with $p=40$ and $n=1000$ is 8 seconds for the top-down and bottom-up
	methods, but 4,500 seconds for the greedy search.

	\subsection{High-dimensional Setting}\label{sec:highDSimulations}
	We now test the proposed procedures in a high-dimensional setting with
	$p > n$ in two scenarios. 
	Random
	DAGs with $p$ nodes and a unique topological ordering are generated
	by: (1) always including edge $v \rightarrow v+1$ for $v < p$, 
	and  either (2a) for each $v>2$, including $u_1,u_2\to v$, 
	where 
	$u_i<v$, and $u_i$ has out-degree $d_{\text{out}}(u_i)<4$, or 
	(2b) for each $v>2$, including $u_1,u_2\to v$, 
	where
	$u_i<\min(v,10) $. 
	In both scenarios, the maximum in-degree is fixed to be $q=3$.
	In the first scenario, it is also guaranteed that 
	the maximum Markov blanket size is small, bounded by $k\leq 15$.  
	In the second scenario when there exists hubs in the graph, 
	the maximum Markov blanket size grows with $p$, 
	with $k\geq 0.2p$. The errors are standard normal.

	Algorithm~\ref{alg:source} with \eqref{eq:sourceHD} as HTD
        (high-dimensional top-down) and to the bottom-up method of
        \citet{Ghoshal2018} as HBU.  The best subset search step in
        HTD is carried with subset size $q=3$; increasing $q$ beyond
        the true maximum in-degree does not change performance
        substantially.  The HBU is tuned with
        $\lambda_n = 0.5\sqrt{\log(p)/n}$.  Results for greedy search
        are not shown as computation becomes intractable when
        $p>100$. Performance is measured by Kendall's $\tau$ to
        provide direct comparison.

	\begin{table}[h]
		\centering
		\caption{High-dimensional setting with maximum in-degree $q = 3$}
		\label{tab:rand_high}
		\begin{tabular}{|c|c|c|c|c|c|c|}
			\hline
			&       &  \multicolumn{2}{c|}{Small $k$} &\multicolumn{2}{c|}{Hub graph}  \\ \hline
			$n$                    & $p$             & HTD& HBU  & HTD& HBU        \\ \hline \hline
			\multirow{5}{*}{80}  & 0.5$n$ & 0.99  & 0.89 & 1.00  & 0.70\\ \cline{2-6}
			& 0.75$n$          & 			0.98  & 0.89 & 0.99  & 0.52\\ \cline{2-6}
			& $n$              & 			0.95  & 0.87 & 0.95  & 0.39\\ \cline{2-6} 
			& 1.5$n$          & 			0.84  & 0.83 & 0.77  & 0.25\\ \cline{2-6}
			& 2$n$            & 			0.72  & 0.73 & 0.55  & 0.16 \\ \hline \hline
			\multirow{5}{*}{100} & 0.5$n$ & 1.00  & 0.93 & 1.00  & 0.70\\ \cline{2-6}
			& 0.75$n$          & 			0.99  & 0.92 & 1.00  & 0.50\\ \cline{2-6}
			& $n$               & 			0.97  & 0.87 & 0.97  & 0.38\\ \cline{2-6}
			& 1.5$n$         & 				0.86  & 0.84 & 0.74  & 0.26\\ \cline{2-6}
			& 2$n$           &  			0.73  & 0.78 & 0.63  & 0.12\\ \hline \hline
			\multirow{5}{*}{200} & 0.5$n$ & 1.00  & 0.95 & 1.00  & 0.77\\ \cline{2-6}
			& 0.75$n$           & 			1.00  & 0.90 & 1.00  & 0.61\\ \cline{2-6}
			& $n$          & 				0.99  & 0.79 & 0.99  & 0.48\\ \cline{2-6}
			& 1.5$n$            & 			0.87  & 0.74 & 0.80  & 0.20\\ \cline{2-6}
			& 2$n$            & 		0.74	  & 0.64 & 0.65  & 0.13\\ \hline
		\end{tabular}
	\end{table}

	Table~\ref{tab:rand_high} demonstrates that in the first scenario,
	both methods perform reasonably well when the considered graph has small Markov blanket. 
	The HTD procedure performs the best in  low-dimensional and moderately high-dimensional settings, and both methods have similar performance in very high-dimensional settings. 
	However, when there exists nodes with very large Markov blanket, the top-down method substantially outperforms the bottom-up method.  
	
	On our personal computer, the average run time for problems of size $p=200$
	is 10 minutes for the HTD method with $q=3$. 
	The computational complexity of HBU
	is determined by the choice of tunning parameter in the
	precisions estimation step. 
	
	Additional simulation settings are presented in
        Appendix~\ref{sec:addSim}-\ref{sec:addSim4} in the supplement including a setting with Rademacher errors as considered by \citet{Ghoshal2018}.

	\section{Discussion}
	In this note, we proposed a simple method for causal discovery
        under a linear structural equation model with equal error
        variances. The procedure consistently estimates a topological
        ordering of the underlying graph and easily extends to the
        high-dimensional setting where $p > n$.  Simulations
        demonstrate that the procedure is an attractive alternative to
        previously considered greedy search
        methods 
	in terms of both accuracy and computational effort.  The
        advantages of the proposed procedures become especially
        salient as the number of considered variables increases.

        In comparison to the related work of \cite{Ghoshal2018}, our
        approach is computationally more demanding for graphs with
        higher in-degree but requires only control over the maximum
        in-degree of the graph as opposed to the maximum degree.
        We also note that as shown in simulations in
        Appendix~\ref{sec:addSim4} a hybrid method in which greedy
        search is initialized at estimates obtained from our variance
        ordering procedures can yield further improvements in
        performance.

        Finally, we note that all discussed methods extend
	to structural equation models where the error variances are unequal, but known up to ratio.
	Indeed, if $\VV(\varepsilon_j)=a_j^2\sigma^2$ for some unknown
	$\sigma^2$ but known $a_1,\dots,a_p$, we may consider
	$\tilde X_j= X_j/a_j$ instead of the original variables.

	\section*{Acknowledgements}
	
	This work was supported by the U.S.\ National Science Foundation (Grant No.~DMS 1712535).

	\bibliographystyle{apalike}
	\bibliography{bib}
	
	\newpage
	\appendix
        \thispagestyle{empty}
        \setcounter{page}{1}

        \noindent
        \begin{center}
          {\bf\large Supplementary material for On Causal Discovery \\with Equal Variance Assumption}
        \end{center}
	\section{Proof of
		Theorem~\ref{thm:topDownProof}}
	\label{sec:error-prop-proofs}
	
	We first give a lemma that addresses the estimation error for inverse
	covariances.
	
	\begin{lemma}\label{lem:inverse}
		Assume $X \sim (B, \sigma^2,\gamma)$.  Suppose all $(q + 1) \times (q + 1)$
		principal submatrices of $\Sigma = \EE(XX^T)$ have minimum
		eigenvalue at least $\lambda_{\min} > 0$.  If for  $\epsilon,\eta>0$
		we have
		\begin{equation}\label{eq:inversionError}
		n\;\geq\;
		(q+1)^2\left\{\log(p^2+p)-\log\left(\epsilon/2\right)\right\}128\left(1+4\frac{\gamma^2}{\sigma^2}\right)^2
		\left(\max_{j\in V}
		\Sigma_{j,j}\right)^2\left(\frac{ \eta \lambda_{\min} + 1}{\eta\lambda_{\min}^2  }  \right)^2.
		\end{equation}
		then
		$$\displaystyle\max_{C\subseteq V,|C|\le q+1}\;\Vert(\Sigma_{C,C})^{-1} - (\hat
		\Sigma_{C,C})^{-1}\Vert\infty \leq \eta$$ with probability at least
		$1-\epsilon$.
		%
	\end{lemma}
	\begin{proof}
		Let
		$\delta = \frac{\eta\lambda_{\min}^2 }{(q+1)( \eta \lambda_{\min} +
			1)}$.  Because $\delta < \frac{\lambda_{\min}}{q+1}$, by Lemma 5 from
		\citet{Harris2013}, we have
		\[
		\max_{C\subseteq V,|C|\le (q+1)}\;\Vert(\Sigma_{C,C})^{-1} - (\hat
		\Sigma_{C,C})^{-1}\Vert_\infty \;\leq\;
		\frac{(q+1)\delta/\lambda_{\min}^2}{1-(q+1)\delta/\lambda_{\min}} \;= \;\eta
		\]
		provided $\|\Sigmah- \Sigma\|_\infty\le \delta$.  The proof is thus
		complete if we show that
		$\PP{\Vert\Sigmah- \Sigma\Vert_\infty> \delta} \leq \epsilon$.
		
		Note that
		$X_j = \varepsilon_j + \sum_{k \in \AN(j)}\pi_{jk}\varepsilon_k$ has
		variance $\sigma^2(1 + \sum_{k \in \AN(j)}\pi_{jk}^2)$.  Since
		$\gamma$ is a bound on the sub-Gaussian parameters of all
		$\epsilon_l$, it follows that
		$X_j / \sqrt{\VV(X_j)}$ is sub-Gaussian with parameter at most
		$\gamma/\sigma$.  Lemma 1 of \cite{ravikumar2011} applies and gives
		\[
		\mathbb{P}\{|\Sigmah_{i,j}-\Sigma_{i,j}|> \delta\}\leq 4\exp\left\{-\frac{n
			\delta^2 }{128(1+4\gamma^2/\sigma^2)^2 \max_j
			(\Sigma_{j,j})^2}\right\}\leq \frac{2}{p(p+1)}\epsilon.
		\]
		A union bound over the entries of $\Sigma$ yields that indeed
		$\PP\left(\Vert\Sigmah- \Sigma\Vert_\infty> \delta\right) \leq \epsilon$.
	\end{proof}

	\begin{proof}[of Theorem~\ref{thm:topDownProof}]
		Our assumption on $n$ is as in~(\ref{eq:inversionError}) with $\eta = \zeta/(2 \sigma^2 )$.
		Lemma~\ref{lem:inverse} thus implies that, with probability at least
		$1- \epsilon$, we have for all subsets $\Theta\subseteq V$ with $|\Theta | < q+1$ that
		\begin{equation}
		\label{eq:precision-error}
		\Vert (\hat \Sigma_{\Theta, \Theta})^{-1}- (\Sigma_{\Theta, \Theta})^{-1}\Vert_\infty\leq  \frac{\zeta}{2\sigma^2}.
		\end{equation}
		
		Let $j$ be a source in $\G(B)$, and
		let $k$ be a non-source.  Note that variance of $j$ conditional on some set $C_1$ is
		\[\sigma^2_{j\mid C_1} = \frac{1}{\left\{(\Sigma_{C_1 \cup \{j\},
			C_1 \cup \{j\}})^{-1}\right\}_{j,j}}.\]
		By Lemma~\ref{lem:topDownHD}, for any $C_1,C_2 \subseteq \Theta \subseteq V \setminus \{j,k\}$ such that $\Theta$ is an ancestral set and $\PA(j) \subseteq C_1$ 
		
		\begin{equation}
		\begin{aligned}
		\left\{(\Sigma_{C_1 \cup \{j\},
			C_1 \cup \{j\}})^{-1}\right\}_{j,j} - \left\{(\Sigma_{C_2 \cup \{k\},
			C_2 \cup \{k\}})^{-1}\right\}_{k,k} &\geq \frac{1}{\sigma^2} - \frac{1}{\sigma^2(1 + \zeta)} 
		 \geq \frac{\zeta}{\sigma^2}
		\end{aligned}
		\end{equation}
		Using~(\ref{eq:precision-error}), when $|C_1|$ and $|C_2|$ are both at most $q$, we obtain that
		\begin{equation}
		\left\{(\hat \Sigma_{C_1 \cup \{j\},
			C_1 \cup \{j\}})^{-1}\right\}_{j,j} - \left\{(\hat \Sigma_{C_2 \cup \{k\},
			C_2 \cup \{k\}})^{-1}\right\}_{k,k}
		- \frac{\zeta}{\sigma^2} \;>\;0.
		\end{equation}
		Thus
		$\hat \sigma^2_{j\mid C_1} -  \hat \sigma^2_{k\mid C_2} > 0$ which implies that  Algorithm~\ref{alg:source} correctly selects a source node at each step. On the first step, $\Theta = \emptyset$ which is trivially an ancestral set. By induction, each subsequent step then correctly adds a sink to $\Theta$ so $\Theta$ remains ancestral and a correct ordering is recovered.     
	\end{proof}

	\section{Simulations as in \citet{Peters2014}} \label{sec:addSim}
	We revisit the simulation study of \citet{Peters2014}. DAGs are generated by first creating a random topological ordering, then between any two nodes, an edge is included with probability $p_c$. We simulate a sparse setting with $p_c=3/(2p-2)$ and a dense setting with $p_c=0.3$. The linear coefficients are drawn uniformly from  $[-1, -.1]\cup[.1, 1]$ and the errors are drawn from a standard Gaussian distribution.
	Since there may not be a unique ordering for the true graph, we compute the Hamming distance between the true and estimated adjacency matrix rather than Kendall's $\tau$.

	Tables~\ref{tab:DAG_dense} and~\ref{tab:DAG_sparse} demonstrate that 
	in both settings, the greedy algorithm performs better when $p$ is small. However, when $p = 40$ the proposed algorithms infer
	the graph more accurately. In the dense setting, the proposed methods have similar FDR to greedy search, but substantially higher recall.
	In the sparse setting, the proposed methods have lower recall than greedy search, but also substantially lower FDR.

	\begin{table}[h]
		\centering
		\caption{Dense setting}
		\label{tab:DAG_dense}
		\begin{tabular}{ll|c|c|c|l|l|l|c|c|c|c|c|c|}
			\cline{3-14}
			&      & \multicolumn{3}{c|}{Hamming Dist.} & \multicolumn{3}{c|}{Recall \%} & \multicolumn{3}{c|}{Flipped \%} & \multicolumn{3}{c|}{FDR \%} \\ \hline
			\multicolumn{1}{|l|}{p}                   & n    & TD         & BU        & GDS       & TD       & BU       & GDS      & TD          & BU          & GDS        & TD      & BU      & GDS     \\ \hline \hline
			\multicolumn{1}{|l|}{\multirow{3}{*}{5}}  & 100  & 1.3        & 1.3       & 1.1       & 73       & 73     & 78     & 7        & 7        & 7       & 16      & 15      & 18      \\ \cline{2-14}
			\multicolumn{1}{|l|}{}                    & 500  & 0.7        & 0.7       & 0.5       & 80     & 80     & 88     & 4        & 4        & 5       & 8       & 7       & 9       \\ \cline{2-14}
			\multicolumn{1}{|l|}{}                    & 1000 & 0.5        & 0.5       & 0.4       & 85    & 84     & 92     & 3        & 3        & 5       & 5       & 5       & 7       \\ \hline \hline
			\multicolumn{1}{|l|}{\multirow{3}{*}{20}} & 100  & 31         & 32        & 30        & 73     & 73     & 74     & 4        & 3        & 6       & 27      & 28      & 25      \\ \cline{2-14}
			\multicolumn{1}{|l|}{}                    & 500  & 22         & 22        & 14        & 91    & 91     & 91    & 2        & 3        & 4       & 24      & 24      & 13      \\ \cline{2-14}
			\multicolumn{1}{|l|}{}                    & 1000 & 28         & 28        & 8         & 94     & 94     & 96     & 2        & 2        & 2       & 21      & 21      & 10      \\ \hline \hline
			\multicolumn{1}{|l|}{\multirow{3}{*}{40}} & 100  & 170        & 174       & 215       & 66     & 65     & 54     & 2        & 3        & 8       & 36      & 37      & 45      \\ \cline{2-14}
			\multicolumn{1}{|l|}{}                    & 500  & 152        & 155       & 186       & 93     & 93     & 76     & 2        & 2        & 9       & 38      & 39      & 42      \\ \cline{2-14}
			\multicolumn{1}{|l|}{}                    & 1000 & 136        & 137       & 168       & 96     & 95    & 83     & 1        & 1        &8       & 36      & 36      & 38      \\ \cline{1-14}
		\end{tabular}
	\end{table}
	
	\begin{table}[h]
		\centering
		\caption{Sparse setting}
		\label{tab:DAG_sparse}
		\begin{tabular}{ll|c|c|c|l|l|l|c|c|c|c|c|c|}
			\cline{3-14}
			&      & \multicolumn{3}{c|}{Hamming Dist.} & \multicolumn{3}{c|}{Recall \%} & \multicolumn{3}{c|}{Flipped \%} & \multicolumn{3}{c|}{FDR \%} \\ \hline
			\multicolumn{1}{|l|}{p}                   & n    & TD         & BU        & GDS       & TD       & BU       & GDS      & TD          & BU         & GDS        & TD      & BU      & GDS     \\ \hline \hline
			\multicolumn{1}{|l|}{\multirow{3}{*}{5}}  & 100  & 1.6        & 1.7       & 1.4       & 74     & 73     & 78     & 8        & 8       & 8       & 18      & 18      & 17      \\ \cline{2-14}
			\multicolumn{1}{|l|}{}                    & 500  & 0.8        & 0.9       & 0.6       & 85     & 84    & 91     & 3        & 4       & 5       & 7       & 7       & 9       \\ \cline{2-14}
			\multicolumn{1}{|l|}{}                    & 1000 & 0.6        & 0.6       & 0.4       & 88     & 88     & 94     & 3        & 4       & 5       & 6       & 6       & 7       \\ \hline \hline
			\multicolumn{1}{|l|}{\multirow{3}{*}{20}} & 100  & 7          & 7         & 12        & 69     & 69     & 81     & 4        & 4       & 6       & 16      & 17      & 43      \\ \cline{2-14}
			\multicolumn{1}{|l|}{}                    & 500  & 3.5        & 3.5       & 4.5       & 85     & 84     & 93     & 4        & 4       & 4       & 9       & 8       & 21      \\ \cline{2-14}
			\multicolumn{1}{|l|}{}                    & 1000 & 2.2        & 2.2       & 2.8       & 90     & 90     & 97     & 3        & 2       & 3       & 5       & 5       & 14      \\ \hline \hline
			\multicolumn{1}{|l|}{\multirow{3}{*}{40}} & 100  & 14         & 15        & 45        & 64     & 63     & 78     & 3        & 4       & 8       & 16      & 18      & 62      \\ \cline{2-14}
			\multicolumn{1}{|l|}{}                    & 500  & 7          & 7         & 16        & 84     & 84     & 94     & 3        & 3       & 3       & 8       & 7       & 33      \\ \cline{2-14}
			\multicolumn{1}{|l|}{}                    & 1000 & 5          & 5         & 10        & 90     & 89     & 97     & 3        & 3       & 3       & 6       & 6       & 24      \\ \cline{1-14}
		\end{tabular}
	\end{table}

        \section{Simulations as in \citet{Ghoshal2018}}
        \label{sec:addSim2}
We construct random graphs as in Section~\ref{sec:highDSimulations}, but we follow the data sampling procedure as used in \citet{Ghoshal2018}. All linear coefficients are
drawn uniformly from $\pm[.5, 1]$, and errors
are drown from the Rademacher distribution and scaled to have $\sigma_i^2=0.8$. 
Table~\ref{tab:rade_high} demonstrates that 
both methods performs reasonably well when Markov blankets
are restricted to be small,
and the top-down approach
performs substantially better when there are hubs. 

\begin{table}[h]
	\centering
	\caption{High-dimensional setting with Rademacher noise and maximum in-degree $q = 3$}
	\label{tab:rade_high}
	\begin{tabular}{|c|c|c|c|c|c|}
		\hline
		&       &  \multicolumn{2}{c|}{Small $k$} &\multicolumn{2}{c|}{Hub graph}  \\ \hline
		$n$                    & $p$             & HTD & HBU & HTD & HBU          \\ \hline \hline
		\multirow{5}{*}{80}  & 0.5$n$ & 0.99 & 0.95 & 0.98 & 0.73\\ \cline{2-6}
		& 0.75$n$          & 			0.98 & 0.90 & 0.89 & 0.46\\ \cline{2-6}
		& $n$              & 			0.96 & 0.90 & 0.76 & 0.36\\ \cline{2-6} 
		& 1.5$n$          & 			0.84 & 0.86 & 0.52 & 0.23\\ \cline{2-6}
		& 2$n$            & 			0.71 & 0.80 & 0.35 & 0.10 \\ \hline \hline
		\multirow{5}{*}{100} & 0.5$n$ & 0.99 & 0.97 & 0.99 & 0.69\\ \cline{2-6}
		& 0.75$n$          & 			0.99 & 0.95 & 0.92 & 0.46\\ \cline{2-6}
		& $n$               & 			0.96 & 0.93 & 0.76 & 0.34\\ \cline{2-6}
		& 1.5$n$         & 				0.84 & 0.88 & 0.52 & 0.26\\ \cline{2-6}
		& 2$n$           &  			0.72 & 0.82 & 0.39 & 0.13\\ \hline \hline
		\multirow{5}{*}{200} & 0.5$n$ & 1.00 & 0.99 & 1.00 & 0.79\\ \cline{2-6}
		& 0.75$n$           & 			1.00 & 0.98 & 0.98 & 0.59\\ \cline{2-6}
		& $n$          & 				0.98 & 0.97 & 0.86 & 0.47\\ \cline{2-6}
		& 1.5$n$            & 			0.86 & 0.84 & 0.61 & 0.20\\ \cline{2-6}
		& 2$n$             & 			0.73 & 0.77 & 0.48 & 0.10 \\ \hline
	\end{tabular}
\end{table}

\section{Simulations of fully connected graphs}
        \label{sec:addSim3}
We run simulations with fully connected graphs, as suggested by a
reviewer. The linear coefficients are drawn uniformly from
$\pm[.3, 1]$ and the errors are drawn from a standard Gaussian
distribution.  The results confirm the advantages of the proposed
methods and are shown in Table~\ref{tab:fully_connected}.  In general,
the estimated graphs from the top-down and bottom-up procedure differ
only slightly, and the values reported in the table differ in the 3rd
or 4th digit.

	\begin{table}[h]
	\centering
	\caption{Fully connected setting}
	\label{tab:fully_connected}
	\begin{tabular}{ll|c|c|c|c|c|c|c|c|c|c|c|c|}
		\cline{3-14}
		&      & \multicolumn{3}{c|}{Kendall's $\tau$} & \multicolumn{3}{c|}{Recall \%} & \multicolumn{3}{c|}{Flipped \% } & \multicolumn{3}{c|}{FDR \%} \\ \hline
		\multicolumn{1}{|l|}{$p$}                   & $n$    & TD& BU     & GDS        & TD& BU       & GDS       & TD& BU       & GDS        & TD& BU     & GDS   \\ \hline \hline
		\multicolumn{1}{|l|}{\multirow{3}{*}{5}}  & 100 & 0.92 & 0.93 & 0.83 & 91 & 92 & 80 & 4 & 3 & 7 & 4 & 4 & 9 \\  \cline{2-14}
		\multicolumn{1}{|l|}{}                    & 500  & 0.99 & 0.99 & 0.97 & 98 & 98 & 98 & 1 & 1 & 1 & 1 & 1 & 1 \\ \cline{2-14}
		\multicolumn{1}{|l|}{}                    & 1000 & 1.00 & 1.00 & 0.99 & 99 & 100 & 99 & 0 & 0 & 1 & 0 & 0 & 1 \\     \hline \hline
		\multicolumn{1}{|l|}{\multirow{3}{*}{20}} & 100  & 0.98 & 0.98 & 0.62 & 74 & 74 & 45 & 1 & 1 & 9 & 1 & 1 & 17 \\ \cline{2-14}
		\multicolumn{1}{|l|}{}                    & 500  & 1.00 & 1.00 & 0.73 & 90 & 90 & 66 & 0 & 0 & 8 & 0 & 0 & 12 \\ \cline{2-14}
		\multicolumn{1}{|l|}{}                    & 1000 & 1.00 & 1.00 & 0.81 & 92 & 92 & 76 & 0 & 0 & 7 & 0 & 0 & 8 \\ \hline \hline
		\multicolumn{1}{|l|}{\multirow{3}{*}{40}} & 100  & 0.99 & 0.99 & 0.55 & 42 & 42 & 33 & 0 & 0 & 7 & 1 & 1 & 17 \\  \cline{2-14}
		\multicolumn{1}{|l|}{}                    & 500  & 1.00 & 1.00 & 0.62 & 50 & 50 & 49 & 0 & 0 & 8 & 0 & 0 & 14 \\  \cline{2-14}
		\multicolumn{1}{|l|}{}                    & 1000  & 1.00 & 1.00 & 0.67 & 52 & 52 & 59 & 0 & 0 & 8 & 0 & 0 & 12 \\  \hline
	\end{tabular}
\end{table}

\section{As initializer for  greedy search}
        \label{sec:addSim4}
As suggested by a reviewer, we explore the performance of the greedy DAG search (GDS) 
algorithm initialized with the estimates from the proposed procedures. 
We run simulations with the same data as in Section~\ref{sec:low-dim-sim}. 
Tables~\ref{tab:combined_dense}  and~\ref{tab:combined_sparse}  show averages over 500 random 
realizations for the top-down procedure (TD), 
the greedy DAG search with random initialization (GR), and the greedy DAG search with warm initialization  (GW). 
The GR procedure is identical to the GDS procedure described in Section~\ref{sec:low-dim-sim}
and \citet{Peters2014}. 
In the GW procedure, we initialize with the output from the top-down
method,
then search through a large number of graph neighbors ($k=300$) at 
each greedy step.
Since the GW procedure is supplied with a good initializer,
we do not restart the greedy search
after it terminates, 
while 5 random restarting with $k=p,2p,3p,5p,300$ is used in GR to insure performance. 
For simplicity, we omitted the experiment with the  bottom-up procedure (BU).

Tables~\ref{tab:combined_dense} and~\ref{tab:combined_sparse} 
shows that 
in all the settings, GW performs better than the other two methods,
especially when $p$ is large. 
For reference, the average run time in the dense setting 
with $p=40$ and $n=1000$ is 8 seconds for the top-down method, 4,500 seconds for GR, and 400 seconds for GW.

	\begin{table}[t]
	\centering
	\caption{Low-dimensional dense settings}
	\label{tab:combined_dense}
	\begin{tabular}{ll|c|c|c|c|c|c|c|c|c|c|c|c|}
		\cline{3-14}
		&      & \multicolumn{3}{c|}{Kendall's $\tau$} & \multicolumn{3}{c|}{Recall \%} & \multicolumn{3}{c|}{Flipped \% } & \multicolumn{3}{c|}{FDR \%} \\ \hline
		\multicolumn{1}{|l|}{$p$}                   & $n$    & TD& GR     & GW        & TD& GR     &  GW          & TD& GR     &   GW            & TD& GR     &  GW    \\ \hline \hline
		\multicolumn{1}{|l|}{\multirow{3}{*}{5}}  & 100 & 0.85 & 0.88 & 0.88 & 91 & 91 & 91 & 7 & 6 & 6 & 17 & 9 & 10 \\  \cline{2-14}
		\multicolumn{1}{|l|}{}                    & 500  & 0.98 & 0.98 & 0.99 & 99 & 99 & 99 & 1 & 1 & 1 & 4 & 2 & 2 \\ \cline{2-14}
		\multicolumn{1}{|l|}{}                    & 1000 & 0.99 & 0.99 & 0.99 & 99 & 99 & 99 & 1 & 1 & 1 & 3 & 1 & 1 \\     \hline \hline
		\multicolumn{1}{|l|}{\multirow{3}{*}{20}} & 100  & 0.92 & 0.61 & 0.94 & 85 & 62 & 90 & 3 & 13 & 3 & 32 & 43 & 15 \\ \cline{2-14}
		\multicolumn{1}{|l|}{}                    & 500  & 0.99 & 0.75 & 0.99 & 99 & 81 & 99 & 1 & 11 & 0 & 28 & 35 & 3 \\ \cline{2-14}
		\multicolumn{1}{|l|}{}                    & 1000 & 1.00 & 0.82 & 1.00 & 100 & 88 & 100 & 0 & 8 & 0 & 26 & 28 & 2 \\ \hline \hline
		\multicolumn{1}{|l|}{\multirow{3}{*}{40}} & 100  & 0.96 & 0.53 & 0.96 & 71 & 44 & 84 & 2 & 11 & 2 & 41 & 58 & 20 \\  \cline{2-14}
		\multicolumn{1}{|l|}{}                    & 500  & 0.99 & 0.59 & 1.00 & 96 & 63 & 100 & 0 & 14 & 0 & 41 & 57 & 4 \\  \cline{2-14}
		\multicolumn{1}{|l|}{}                    & 1000  & 1.00 & 0.64 & 1.00 & 97 & 71 & 100 & 0 & 14 & 0 & 40 & 57 & 2 \\  \hline
	\end{tabular}
\end{table}

	\begin{table}[t]
	\centering
	\caption{Low-dimensional sparse settings}
	\label{tab:combined_sparse}
	\begin{tabular}{ll|c|c|c|c|c|c|c|c|c|c|c|c|}
		\cline{3-14}
		&      & \multicolumn{3}{c|}{Kendall's $\tau$} & \multicolumn{3}{c|}{Recall \%} & \multicolumn{3}{c|}{Flipped \% } & \multicolumn{3}{c|}{FDR \%} \\ \hline
		\multicolumn{1}{|l|}{$p$}                   & $n$    & TD& GR     & GW        & TD& GR     &  GW          & TD& GR     &   GW            & TD& GR     &  GW    \\ \hline \hline
		\multicolumn{1}{|l|}{\multirow{3}{*}{5}}  & 100 & 0.87 & 0.88 & 0.87 & 91 & 90 & 91 & 6 & 6 & 6 & 16 & 9 & 10 \\  \cline{2-14}
		\multicolumn{1}{|l|}{}                    & 500  & 0.98 & 0.98 & 0.98 & 98 & 99 & 99 & 1 & 1 & 1 & 5 & 2 & 2 \\ \cline{2-14}
		\multicolumn{1}{|l|}{}                    & 1000 & 0.99 & 0.99 & 0.99 & 99 & 99 & 99 & 1 & 1 & 1 & 3 & 1 & 1 \\     \hline \hline
		\multicolumn{1}{|l|}{\multirow{3}{*}{20}} & 100  & 0.77 & 0.60 & 0.82 & 85 & 77 & 90 & 9 & 15 & 7 & 35 & 39 & 25 \\ \cline{2-14}
		\multicolumn{1}{|l|}{}                    & 500  & 0.96 & 0.77 & 0.98 & 98 & 89 & 99 & 2 & 10 & 1 & 19 & 26 & 8 \\ \cline{2-14}
		\multicolumn{1}{|l|}{}                    & 1000 & 0.99 & 0.81 & 0.99 & 100 & 90 & 100 & 0 & 9 & 0 & 14 & 23 & 4 \\ \hline \hline
		\multicolumn{1}{|l|}{\multirow{3}{*}{40}} & 100  & 0.72 & 0.47 & 0.79 & 81 & 72 & 89 & 10 & 20 & 7 & 38 & 54 & 36 \\  \cline{2-14}
		\multicolumn{1}{|l|}{}                    & 500  & 0.96 & 0.58 & 0.98 & 98 & 81 & 99 & 2 & 18 & 1 & 24 & 47 & 13 \\  \cline{2-14}
		\multicolumn{1}{|l|}{}                    & 1000  & 0.99 & 0.61 & 0.99 & 99 & 82 & 100 & 1 & 17 & 0 & 17 & 48 & 8 \\  \hline
	\end{tabular}
\end{table}

\end{document}